\documentclass{llncs}
\usepackage[T1]{fontenc}
\usepackage[utf8x]{inputenx}

\usepackage{graphicx}
\usepackage{xspace}
\usepackage{amsmath}
\usepackage{amssymb}
\usepackage{hyperref}
\usepackage{cleveref}
\usepackage{enumitem}
\usepackage{twoopt}
\usepackage{etoolbox}
\usepackage{proof}
\usepackage{multicol}
\usepackage{color}
\usepackage{anyfontsize}

\Crefname{equation}{Eq.}{Eqs.}
\crefname{equation}{eq.}{eqs.}
\Crefname{figure}{Fig.}{Figs.}
\crefname{figure}{fig.}{figs.}
\Crefname{tabular}{Tab.}{Tabs.}
\crefname{tabular}{tab.}{tabs.}
\Crefname{definition}{Def.}{Defs.}
\crefname{definition}{def.}{defs.}
\Crefname{proposition}{Prop.}{Props.}
\crefname{proposition}{prop.}{props.}
\Crefname{section}{Sect.}{Sects.}
\crefname{section}{sect.}{sects.}

\usepackage{pifont}
\newcommand{\cmark}{\ding{51}}
\newcommand{\xmark}{\ding{55}}

\usepackage{tikz}
\usetikzlibrary{calc,automata,positioning}
\tikzstyle{every node}=[font=\scriptsize]
\tikzstyle{state} = [draw,fill=white,circle,minimum size=4mm,inner sep=0pt,thick]
\tikzstyle{dot} = [fill,circle,inner sep=0mm,minimum size=1.25mm,line width=0mm]

\setitemize{noitemsep,topsep=0pt,parsep=0pt,partopsep=0pt,leftmargin=9.5pt,labelsep=3pt}
\setenumerate{noitemsep,topsep=0pt,parsep=0pt,partopsep=0pt}

\newcommand{\modest}{\textsc{\mbox{Modest}}\xspace}

\newcommand{\toolset}{\textsc{\mbox{Modest} Toolset}\xspace}

\newcommand{\eg}{e.g.\xspace}
\newcommand{\ie}{i.e.\xspace}
\newcommand{\st}{\ifmmode \:\text{s.t.}\ \else s.t.\xspace\fi}
\newcommand{\etal}{et al.\xspace}
\newcommand{\wrt}{w.r.t.\xspace}

\newcommand{\set}[1]{\ensuremath{\{\,#1\,\}}}
\newcommand{\tuple}[1]{\ensuremath{\langle #1 \rangle}}
\newcommand{\powerset}[1]{\ensuremath{\mathcal{P}({#1})}\xspace}
\newcommand{\RR}{\ensuremath{\mathbb{R}}\xspace}
\newcommand{\RRplus}{\ensuremath{\mathbb{R}^+}\xspace}
\newcommand{\RRpluszero}{\ensuremath{\mathbb{R}^{+}_{0}}\xspace}

\newcommand{\Dist}[1]{\ensuremath{\mathrm{Dist}({#1})}\xspace} 
\newcommand{\support}[1]{\ensuremath{\mathrm{support}({#1})}\xspace}
\newcommand{\Dirac}[1]{\ensuremath{\mathcal{D}(#1)}\xspace}
\newcommand{\Prob}[1]{\ensuremath{\mathrm{Prob}({#1})}} 

\newcommand{\defeq}{\mathrel{\vbox{\offinterlineskip\ialign{\hfil##\hfil\cr{\tiny \ensuremath{\mathrm{def}}}\cr\noalign{\kern0.25ex}$=$\cr}}}}
\newcommand{\Clocks}{\ensuremath{\mathcal{C}}\xspace}
\newcommand{\ClockDistrs}{\ensuremath{F}\xspace}
\newcommand{\Loc}{\ensuremath{\mathit{Loc}}\xspace}
\newcommand{\Edges}{\ensuremath{E}\xspace}
\newcommand{\Alphabet}{\ensuremath{A}\xspace}
\newcommand{\InitialLoc}[1][]{\ifstrempty{#1}{\ensuremath{\ell_\mathit{init}}}{\ensuremath{\ell_{\mathit{init}_#1}}}\xspace}
\newcommand{\xtr}[1]{\xrightarrow{\protect{\raisebox{-1pt}[0pt][0pt]{\ensuremath{\scriptstyle{#1}}}}}}
\newcommand{\States}{\ensuremath{S}\xspace}
\newcommand{\InitialState}[1][]{\ifstrempty{#1}{\ensuremath{s_\mathit{init}}}{\ensuremath{s_{\mathit{init}_#1}}}\xspace}
\newcommand{\Trans}{\ensuremath{T}\xspace}
\newcommand{\Val}[1][]{\ifstrempty{#1}{\ensuremath{\mathit{Val}}}{\ensuremath{\mathrm{Val}(#1)}}\xspace}
\newcommand{\zeroval}[0]{\ensuremath{\mathbf{0}}}

\newcommand{\Sched}{\ensuremath{\mathfrak{S}}\xspace} 
\newcommand{\sched}{\ensuremath{\mathfrak{s}}\xspace} 
\newcommand{\Pmax}[2]{\ensuremath{\mathrm{P}^{#2}_\mathrm{\!max}(#1)}\xspace}
\newcommand{\Pmin}[2]{\ensuremath{\mathrm{P}^{#2}_\mathrm{\!min}(#1)}\xspace}

\newcommand{\sem}[1]{\ensuremath{[\hspace{-1.5pt}[#1]\hspace{-1.5pt}]}}
\newcommand{\Reach}{\ensuremath{J}\xspace} 
\newcommand{\restart}{R} 
\newcommand{\isfun}{\ensuremath{\colon}} 
\newcommand{\hide}[1]{\ignorespaces} 
\newcommand{\schedngeq}{\ensuremath{\not\succcurlyeq}\xspace}
\newcommand{\schedgeq}{\ensuremath{\succcurlyeq}\xspace}
\newcommand{\schedgtr}{\ensuremath{\succ}\xspace}
\newcommand{\schedeq}{\ensuremath{\approx}\xspace}
\newcommand{\schedneq}{\ensuremath{\not\approx}\xspace}
\newcommand{\schednleq}{\ensuremath{\not\preccurlyeq}\xspace}

\newcommand{\schedlss}{\ensuremath{\prec}\xspace}

\newcommand{\proofnegspace}{\vspace{0pt}}

\hypersetup{
  pdftitle = {A Hierarchy of Scheduler Classes for Stochastic Automata}
}

\begin{document}

\title{
A Hierarchy of Scheduler Classes\\for Stochastic Automata
\thanks{This work is supported by the 3TU project ``Big Software on the Run'', by the JST ERATO HASUO Metamathematics for Systems Design project (JPMJER1603), by the ERC Grant 695614 (POWVER), and by SeCyT-UNC 05/BP12 and 05/B497.}%
}

\author{
\mbox{Pedro R.\ D'Argenio\inst{1,2}
\and Marcus Gerhold\inst{3}
\and Arnd Hartmanns\inst{3}
\and Sean Sedwards\inst{4}}
}

\institute{
Universidad Nacional de C\'ordoba, C\'ordoba, Argentina
\and Saarland University, Saarbr\"ucken, Germany
\and University of Twente, Enschede, The Netherlands
\and National Institute of Informatics, Tokyo, Japan
}

\maketitle

\begin{abstract}
Stochastic automata are a formal compositional model for concurrent stochastic timed systems, with general distributions and nondeterministic choices.
Measures of interest are defined over \emph{schedulers} that resolve the nondeterminism.
In this paper we investigate the power of various theoretically and practically motivated classes of schedulers, considering the classic complete-information view and a restriction to non-prophetic schedulers.
We prove a hierarchy of scheduler classes \wrt unbounded probabilistic reachability.
We find that, unlike Markovian formalisms, stochastic automata distinguish most classes even in this basic setting.
Verification and strategy synthesis methods thus face a tradeoff between powerful and efficient classes.
Using lightweight scheduler sampling, we explore this tradeoff and demonstrate the concept of a useful approximative verification technique for stochastic automata.  
\end{abstract}

\section{Introduction}
\label{sec:Introduction}

The need to analyse continuous-time stochastic models arises in many practical contexts, including critical infrastructures~\cite{ACGHHKMR15}, railway engineering~\cite{RS16}, space mission planning~\cite{BGHKNS16}, and security~\cite{HKKS16}.
This has led to a number of discrete event simulation tools, such as those for networking~\cite{NS3Website,Pon93,ZBG98}, whose probabilistic semantics is founded on generalised semi-Markov processes (GSMP~\cite{HS87,Mat62}).
Nondeterminism arises through inherent concurrency of independent processes~\cite{BKKR11}, but may also be deliberate underspecification.
Modelling such uncertainty with probability is convenient for simulation, but not always adequate~\cite{AY06,KCC05}.
Various models and formalisms have thus been proposed to extend continuous-time stochastic processes with nondeterminism~\cite{BDHK06,BG02,EHZ10,HS00,Her02,Str93}.
It is then possible to \emph{verify} such systems by considering the extremal probabilities of a property.
These are the supremum and infimum of the probabilities of the property in the purely stochastic systems induced by classes of {\em schedulers} (also called \emph{strategies}, \emph{policies} or \emph{adversaries}) that resolve all nondeterminism.
If the nondeterminism is considered controllable, one may alternatively be interested in the \emph{planning} problem of synthesising a scheduler that satisfies certain probability bounds.

We consider closed systems of stochastic automata (SA~\cite{DK05}), which extend GSMP and feature both generally distributed stochastic delays as well as discrete nondeterministic choices.
The latter may arise from non-continuous distributions (\eg deterministic delays), urgent edges, and edges waiting on multiple clocks.
Model checking (extensions of) GSMP is known to be undecidable in general, and numerical verification algorithms exist for very limited subclasses of SA only:
Buchholz \etal\cite{BKS14} restrict to phase-type or matrix-exponential distributions, such that nondeterminism cannot arise (as each edge is guarded by a single clock).
Bryans \etal\cite{BBD03} propose two algorithms that require an a priori fixed scheduler, continuous bounded distributions, and that all active clocks be reset when a location is entered.
The latter forces regeneration on every edge, making it impossible to use clocks as memory between locations.
Regeneration is central to the work of Paolieri, Vicario \etal\cite{BBHPV13}, however they again exclude nondeterminism.
The only approach that handles nondeterminism is the region-based approximation scheme of Kwiatkowska \etal\cite{KNSS00} for a model closely related to SA, but restricted to bounded continuous distributions.
Without that restriction~\cite{HHH14}, error bounds and convergence guarantees are lost.

Evidently, the combination of nondeterminism and continuous probability distributions is a particularly challenging one.
With this paper, we take on the underlying problem from a fundamental perspective:
we investigate the power of, and relationships between, different classes of schedulers for SA.
Our motivation is, on the one hand, that a clear understanding of scheduler classes is crucial to design verification algorithms.
For example, Markov decision process (MDP) model checking works well because memoryless schedulers suffice for reachability, and the efficient time-bounded analysis of continuous-time MDP (CTMDP) exploits a relationship between two scheduler classes that are sufficiently simple, but on their own do not realise the desired extremal probabilities~\cite{BHHK15}.
When it comes to planning problems, on the other hand, practitioners desire \emph{simple} solutions, \ie schedulers that need little information and limited memory, so as to be explainable and suitable for implementation on \eg resource-constrained embedded systems.
Understanding the capabilities of scheduler classes helps decide on the tradeoff between simplicity and the ability to attain optimal results.

We use two perspectives on schedulers from the literature:
the classic com\-plete-information {\em residual lifetimes} semantics~\cite{BD04}, where optimality is defined via his\-tory-dependent schedulers that see the entire current state, and \emph{non-prophetic} schedulers~\cite{HHK16} that cannot observe the timing of \emph{future} events.
Within each perspective, we define classes of schedulers whose views of the state and history are variously restricted (\Cref{sec:Classes}).
We then prove their relative ordering \wrt achieving optimal reachability probabilities (\Cref{sec:PowerOfSchedulers}).
We find that SA distinguish most classes.
In particular, memoryless schedulers suffice in the complete-information setting (as is implicit in the method of Kwiatkowska \etal), but turn out to be suboptimal in the more realistic non-prophetic case.
Considering only the relative order of clock expiration times, as suggested by the first algorithm of Bryans \etal{}, surprisingly leads to partly suboptimal, partly incomparable classes.
Our distinguishing SA are small and employ a common nondeterministic gadget.
They precisely pinpoint the crucial differences and how schedulers interact with the various features of SA, providing deep insights into the formalism itself.

Our study furthermore forms the basis for the application of {\em lightweight scheduler sampling} (LSS) to SA.
LSS is a technique to use Monte Carlo simulation/statistical model checking with nondeterministic models.
On every LSS simulation step, a pseudo-random number generator (PRNG) is re-seeded with a hash of the identifier of the current scheduler and the (restricted) information about the current state (and previous states, for history-dependent schedulers) that the scheduler's class may observe.
The PRNG's first iterate then deterministically decides the scheduler's action.
LSS has been successfully applied to MDP~\cite{DLST15,LST15b,LST15a} and probabilistic timed automata~\cite{DHLS16,DHS17}.
Using only constant memory, LSS samples schedulers uniformly from a selected scheduler class to find ``near-optimal'' schedulers that conservatively approximate the true extremal probabilities.
Its principal advantage is that it is largely indifferent to the size of the state space and of the scheduler space; in general, sampling efficiency depends only on the likelihood of selecting near-optimal schedulers.
However, the mass of {\em near}-optimal schedulers in a scheduler class that also includes the optimal scheduler may be \emph{less} than the mass in a class that does \emph{not} include it.
Given that the mass of optimal schedulers may be vanishingly small, it may be advantageous to sample from a class of less powerful schedulers.
We explore these tradeoffs and demonstrate the concept of LSS for SA in \Cref{sec:Experiments}.
 
\smallskip\noindent\textbf{Other related work.}
Alur \etal first mention nondeterministic stochastic systems similar to SA in~\cite{ACD91}.
Markov automata (MA~\cite{EHZ10}), interactive Markov chains (IMC~\cite{Her02}) and CTMDP are special cases of SA restricted to exponential distributions.
Song \etal\cite{SZG12} look into partial information distributed schedulers for MA, combining earlier works of de Alfaro~\cite{deA99} and Giro and D'Argenio~\cite{GD07} for MDP.
Their focus is on information flow and hiding in parallel specifications.
Wolf \etal\cite{WBM06} investigate the power of classic (time-abstract, deterministic and memoryless) scheduler classes for IMC.
They establish (non-strict) subset relationships for almost all classes \wrt trace distribution equivalence, a very strong measure.
Wolovick and Johr~\cite{WJ06} show that the class of measurable schedulers for CTMDP is complete and sufficient for reachability problems.

\section{Preliminaries}
\label{sec:Preliminaries}

For a given set~$S$, its power set is $\powerset{S}$.
We denote by \RR, $\RRplus$, and $\RRpluszero$ the sets of real numbers, positive real numbers and non-negative real numbers, respectively.
A (discrete) \emph{probability distribution} over a set~$\varOmega$ is a function $\mu \isfun \varOmega \to [0, 1]$, such that $\support{\mu} \defeq \set{\omega \in \varOmega \mid \mu(\omega) > 0}$ is countable and \mbox{$\sum_{\omega \in \support{\mu}}{\mu(\omega)} = 1$}.
$\Dist{\varOmega}$ is the set of probability distributions over~$\varOmega$.
We write $\Dirac{\omega}$ for the \emph{Dirac} distribution for~$\omega$, defined by $\Dirac{\omega}(\omega) = 1$.
$\varOmega$ is \emph{measurable} if it is endowed with a $\sigma$-algebra $\sigma(\varOmega)$: a collection of \emph{measurable} subsets of $\varOmega$.
A (continuous) \emph{probability measure} over $\varOmega$ is a function $\mu \isfun \sigma(\varOmega) \to [0, 1]$, such that $\mu(\varOmega)=1$ and $\mu(\cup_{i \in I}\, B_i) = \sum_{i \in I}\, \mu(B_i)$ for any countable index set~$I$ and pairwise disjoint measurable sets $B_i\subseteq\varOmega$.
$\Prob{\varOmega}$ is the set of probability measures over $\varOmega$.
Each $\mu \isfun \Dist{\varOmega}$ induces a probability measure.
We write $\Dirac{\cdot}$ for the Dirac measure.
Given probability measures $\mu_1$ and $\mu_2$, we denote by $\mu_1 \otimes \mu_2$ the \emph{product measure}: the unique probability measure such that $(\mu_1 \otimes \mu_2)(B_1 \times B_2) = \mu_1(B_1) \cdot \mu_2(B_2)$, for all measurable $B_1$ and $B_2$.
For a collection of measures $(\mu_i)_{i\in I}$, we analogously denote the product measure by $\bigotimes_{i \in I} \mu_i$.
Let $\Val \defeq V \to \RRpluszero$ be the set of valuations for an (implicit) set $V$ of (non-negative real-valued) variables.
$\zeroval \in \Val$ assigns value zero to all variables.
Given $X\subseteq V$ and $v \in \Val$, we write $v[X]$ for the valuation defined by $v[X](x) = 0$ if $x \in X$ and $v[X](y) = v(y)$ otherwise.
For $t \in \RRpluszero$, $v + t$ is the valuation defined by $(v + t)(x) = v(x) + t$ for all $x \in V$.

\subsubsection{Stochastic Automata}
\label{sec:StochasticAutomata}

 extend labelled transition systems with stochastic \emph{clocks}:
real-valued variables that increase synchronously with rate 1 over time and expire some random amount of time after having been \emph{restarted}.
Formally:

\begin{definition}
\label{def:SA}
A \emph{stochastic automaton}~(SA~\cite{DK05}) is a tuple
$\tuple{\Loc, \Clocks, \Alphabet, \Edges, \ClockDistrs, \InitialLoc}$,
where
$\Loc$ is a countable set of \emph{locations},
$\Clocks$ is a finite set of \emph{clocks},
$\Alphabet$ is the finite \emph{action alphabet}, 
and $\Edges \isfun \Loc \to \powerset{ \powerset{\Clocks} \times \Alphabet \times \powerset{\Clocks} \times \Dist{\Loc} }$ is the \emph{edge function}, which maps each location to a finite set of edges that in turn consist of a \emph{guard set} of clocks, a label, a \emph{restart set} of clocks and a distribution over target locations.
$\ClockDistrs \isfun \Clocks \to \Prob{\RRpluszero}$ is the \emph{delay measure function} that maps each clock to a probability measure
, and
$\InitialLoc \in \Loc$ is the \emph{initial location}.
\end{definition}
We also write $\ell \xtr{G, a, R}_\Edges \mu$ for $\tuple{G, a, R, \mu} \in
\Edges(\ell)$.
W.l.o.g. 
we restrict to SA where edges are fully characterised by source state and action label, \ie whenever $\ell \xtr{G_1, a, R_1}_\Edges \mu_1$ and $\ell \xtr{G_2, a, R_2}_\Edges \mu_2$, then $G_1 = G_2$, $R_1 = R_2$ and $\mu_1 = \mu_2$.

Intuitively, an SA starts in $\InitialLoc$ with all clocks expired.
An edge $\ell \xtr{G, a, R}_\Edges \mu$ may be taken only if all clocks in $G$ are expired.
If any edge is enabled, some edge must be taken (\ie all actions are \emph{urgent}).
When an edge is taken, its action is $a$, all clocks in $R$ are restarted, other expired clocks remain expired, and $\mu$ encodes its discrete branching: we move to successor location $\ell'$ with probability $\mu(\ell')$.
There, another edge may be taken immediately or we may need to wait until some further clocks expire, and so on.
When a clock $c$ is restarted, the time until it expires is chosen randomly according to the probability measure $F(c)$.

\begin{example}
\begin{figure}[t]
\begin{minipage}[b]{0.36\textwidth}
\centering
\begin{tikzpicture}[on grid]
  \node[state] (l0) {$\ell_0$};
  \coordinate[left=0.3 of l0.west] (start);
  \node[] (me) [above left=0.4 and 1.1 of l0] {\small$M_0$:};
  \node[] (distr) [above right=0.2 and 1.3 of l0,align=left] {$x\colon \textsc{Uni}(0, 1)$\\$y\colon \textsc{Uni}(0, 1)$};
  \node[state] (l1) [below=1 of l0] {$\ell_1$};
  \node[state] (l2) [below left=0.875 and 0.75 of l1] {$\ell_2$};
  \node[state] (l3) [below right=0.875 and 0.75 of l1] {$\ell_3$};
  \node[state] (yes) [below=1 of l2] {\cmark};
  \node[state] (no) [below=1 of l3] {\xmark};
  ;
  \path[->]
    (start) edge node {} (l0)
    (l0) edge node[right,pos=0.25,inner sep=0.5mm] {\strut$\varnothing$} node[right,pos=0.7,inner sep=0.5mm] {\strut$\text{\restart}(\{ x, y \})$} (l1)
    (l1) edge[] node[left,pos=0.15,inner sep=1mm] {\strut$\varnothing\!$} (l2)
    (l1) edge[] node[right,pos=0.15,inner sep=1mm] {\strut$\varnothing$} (l3)
    (l2) edge node[left,pos=0.25,inner sep=0.5mm] {\strut$\{ x \}$} (yes)
    (l2) edge[bend right=20] node[above,pos=0.33,inner sep=2mm] {\strut$\{ y \}$} (no)
    (l3) edge[bend left=20] node[above,pos=0.33,inner sep=2mm] {\strut$\{ y \}$} (yes)
    (l3) edge node[right,pos=0.25,inner sep=0.5mm] {\strut$\{ x \}$} (no)
  ;
\end{tikzpicture}
\caption{Example SA $M_0$}
\label{fig:SAExample}
\end{minipage}%
\begin{minipage}[b]{0.64\textwidth}
\centering
\begin{tikzpicture}[on grid]
  \node[] (l0) {$\tuple{\ell_0, \tuple{0, 0}, \tuple{0, 0}}$};
  \coordinate[left=0.3 of l0.west] (start);
  \node[] (me) [above left=0.0 and 3.15 of l0] {\small$\sem{M_0}$:};
  \node[] (l1l) [below left=1 and 2.8 of l0] {\strut$\tuple{\ell_1, \tuple{0, 0}, \tuple{0, 0}}$};
  \node[] (l1m) [below=1 of l0] {\strut$\tuple{\ell_1, \tuple{0, 0}, \tuple{e(x), e(y)}}$};
  \node[] (l1r) [below right=1 and 2.8 of l0] {\strut$\tuple{\ell_1, \tuple{0, 0}, \tuple{1, 1}}$};
  \node[] (l1ldots) [below left=0.4 and 0.33 of l0] {$\cdots$};
  \node[] (l1rdots) [below right=0.4 and 0.33 of l0] {$\,\cdots$};
  \node[] (l1mdots1) [below left=1.01 and 1.55 of l0] {$\cdots$};
  \node[] (l1mdots2) [below right=1.01 and 1.61 of l0] {$\cdots$};
  \node[] (l1text) [below right=0.4 and 2.25 of l0,align=left] {$\sim\textsc{Uni}(0, 1)^2$};
  \node[] (l1ldots1) [below left=0.6 and 0.33 of l1l] {$\cdots$};
  \node[] (l1ldots2) [below right=0.6 and 0.33 of l1l] {$\cdots$};
  \node[] (l1rdots1) [below left=0.6 and 0.33 of l1r] {$\,\cdots$};
  \node[] (l1rdots2) [below right=0.6 and 0.33 of l1r] {$\,\cdots$};
  \node[] (l2t) [below left=1.0 and 1.9 of l1m] {$\tuple{\ell_2, \tuple{0, 0}, \tuple{e(x), e(y)}}$};
  \node[] (l3t) [below right=1.0 and 1.9 of l1m] {$\tuple{\ell_3, \tuple{0, 0}, \tuple{e(x), e(y)}}$};
  \node[] (delaytext) [below left=0.5 and 1.9 of l3t,align=left] {\textit{(assume }$e(x) \!<\! e(y)$\textit{)}};
  \node[] (l2b) [below=1 of l2t] {$\tuple{\ell_2, \tuple{e(x), e(x)}, \tuple{e(x), e(y)}}$};
  \node[] (l3b) [below=1 of l3t] {$\tuple{\ell_3, \tuple{e(x), e(x)}, \tuple{e(x), e(y)}}$};
  \node[] (target) [below=0.8 of l2b] {$\tuple{\text{\cmark}, \tuple{e(x), e(x)}, \tuple{e(x), e(y)}}$};
  \node[] (error) [below=0.8 of l3b] {$\tuple{\text{\xmark}, \tuple{e(x), e(x)}, \tuple{e(x), e(y)}}$};
  ;
  \path[->]
    (start) edge node {} (l0)
    (l0) edge[bend left=5] (l1l)
    (l0) edge (l1m)
    (l0) edge[bend right=5] (l1r)
    (l1l) edge (l1ldots1)
    (l1l) edge (l1ldots2)
    (l1r) edge (l1rdots1)
    (l1r) edge (l1rdots2)
    (l1m) edge (l2t)
    (l1m) edge (l3t)
    (l2t) edge node[left] {$e(x)$} (l2b)
    (l3t) edge node[right] {$e(x)$} (l3b)
    (l2b) edge (target)
    (l3b) edge (error)
  ;
\end{tikzpicture}
\caption{Excerpt of the TPTS semantics of $M_0$}
\label{fig:SASemanticsExample}
\end{minipage}
\end{figure}
We show an example SA, $M_0$, in \Cref{fig:SAExample}.
Its initial location is $\ell_0$.
It has two clocks, $x$ and $y$, with $F(x)$ and $F(y)$ both being the continuous uniform distribution over the interval $[0,1]$.
No time can pass in locations $\ell_0$ and $\ell_1$, since they have outgoing edges with empty guard sets.
We omit action labels and assume every edge to have a unique label.
On entering $\ell_1$, both clocks are restarted.
The choice of going to either $\ell_2$ or $\ell_3$ from $\ell_1$ is nondeterministic, since the two edges are always enabled at the same time.
In $\ell_2$, we have to wait until the first of the two clocks expires.
If that is $x$, we have to move to location \cmark; if it is $y$, we have to move to \xmark.
The probability that both expire at the same time is zero.
Location $\ell_3$ behaves analogously, but with the target states interchanged.
\end{example}

\subsubsection{Timed Probabilistic Transition Systems}
form the semantics of SA.
They are finitely-nondeterministic uncountable-state transition systems:

\begin{definition}
\label{def:TPTS}
A (finitely nondeterministic)
\emph{timed probabilistic transition system} (TPTS) is a tuple
$\tuple{\States, \Alphabet', \Trans, \InitialState}$.
$\States$ is a measurable set of states $\tuple{\ell, v, e}\in\Loc\times\Val\times\Val$, where $\ell$ is the current location, $v$ is a valuation assigning to each clock its current value, and $e$ is a valuation keeping track of all clocks' expiration times. 
$\Alphabet' = \RRplus \uplus \Alphabet$ is the \emph{alphabet}, partitioned into \emph{delays} in \RRplus and \emph{jumps} in $\Alphabet$.
$\Trans \isfun \States \to \powerset{\Alphabet' \times \Prob{\States}}$ is the \emph{transition function}, which maps each state to a finite set of transitions, each consisting of a label in $\Alphabet'$ and a measure over target states.
The initial state is $\InitialState \in \States$.
For all $s \in \States$, we require $|\Trans(s)| = 1$ if $\exists\,\tuple{t, \mu} \in \Trans(s) \colon t \in \RRplus$, \ie states admitting delays are deterministic.
\end{definition}
We also write $s \xtr{a}_\Trans \mu$ for $\tuple{a, \mu} \in \Trans(s)$.
A \emph{run} is an infinite alternating sequence $s_0 a_0 s_1 a_1 \!\ldots \in (\States \times \Alphabet')^\omega$, with $s_0 = \InitialState$.
It resolves all nondeterministic and probabilistic choices.
A \emph{scheduler} resolves only the nondeterminism:

\begin{definition}
\label{def:Scheduler}
A measurable function $\sched \isfun (\States \times \Alphabet')^* \times \States \to \Dist{\Alphabet' \times \Prob{\States}}$ is a \emph{scheduler} if, for all histories $h \isfun (\States \times \Alphabet')^* \times \States$, $\tuple{a, \mu} \in \support{\sched(h)}$ implies $\mathit{lst}_h \xtr{a}_\Trans \mu$, where $\mathit{lst}_h$ is the last state of $h$.
\end{definition}
Once a scheduler has chosen $s_i \xtr{a}_\Trans \mu$, the successor state $s_{i+1}$ is picked randomly according to~$\mu$.
Every scheduler $\sched$ defines a probability measure $\mathbb{P}_\sched$ on the space of all runs.
For a formal definition, see~\cite{Wol12}.
As is usual, we restrict to \emph{non-Zeno} schedulers that make time diverge with probability one:
we require $\mathbb{P}_\sched(\Pi_\infty) = 1$, where $\Pi_\infty$ is the set of runs where the sum of delays is $\infty$.
In the remainder of this paper we consider extremal probabilities of reaching a set of goal locations~$G$:
\begin{definition}
\label{def:ReachProb}
For $G \subseteq \Loc$, let $\Reach_G \defeq \set{ \tuple{\ell, v, e} \in \States \mid \ell \in G }$.
For a scheduler \sched, let $\Pi_{\Reach_G}^{\sched}$ be the set of runs under $\sched$ with a state in~$\Reach_G$.
Let $\Sched$ be a class of schedulers.
Then $\Pmin{G}{\Sched}$ and $\Pmax{G}{\Sched}$ are the minimum and maximum \emph{reachability probabilities} for $G$ under $\Sched$, defined as $\Pmin{G}{\Sched} = \inf_{\sched\in\Sched} \mathbb{P}_\sched(\Pi_{\Reach_G})$ and $\Pmax{G}{\Sched} = \sup_{\sched\in\Sched} \mathbb{P}_\sched(\Pi_{\Reach_G})$, respectively.
\end{definition}

\subsubsection{Semantics of Stochastic Automata}
\label{sec:SASemantics}

We present here the residual lifetimes semantics of~\cite{BD04}, simplified for closed SA:
any delay step must be of the minimum delay that makes some edge become enabled.

\begin{definition}
\label{def:SASemantics}
The semantics of an SA $M = \tuple{\Loc, \Clocks, \Alphabet, \Edges, \ClockDistrs, \InitialLoc}$ is the TPTS\\[4pt]
\centerline{$
\sem{M} = \tuple{\Loc \times \Val \times \Val, \Alphabet \uplus \RRplus, \Trans_M, \tuple{\InitialLoc, \zeroval, \zeroval}}
$,}\\[4pt]
where $\Trans_M$ is the smallest \hide{(according to~$<$)} transition function satisfying inference rules\\[5pt]
\centerline{$
\infer
{\tuple{\ell, v, e} \strut\xtr{a}_{\Trans_M} \mu \otimes \Dirac{v[R]} \otimes \mathrm{Sample}^R_e}
{\ell \xtr{G, a, R}_\Edges \mu & \mathrm{En}(G, v, e)}
$}\\[5pt]
\centerline{$\infer
{\tuple{\ell, v, e} \strut\xtr{t}_{\Trans_M} \Dirac{\tuple{\ell, \mathit{v + t}, e}}}
{t \in \RRplus & \exists\hspace{1pt} \ell \xtr{G, a, R}_\Edges \mu \colon \mathrm{En}(G, v + t, e) & \forall\, t' \!\in\! [0, t), \ell \xtr{G, a, R}_\Edges \mu \colon \neg\, \mathrm{En}(G, v + t', e) }
$}\\[5pt]
with $\mathrm{En}(G, v, e) \defeq \forall\,x \in G \colon v(x) \geq e(x)$ characterising the enabled edges and\\[4pt]
\centerline{$
\mathrm{Sample}^R_e \defeq \bigotimes_{c \in \Clocks}
\begin{cases}
F(c) & \text{if } c \in R\\
\Dirac{e(c)} & \text{if } c \notin R.
\end{cases}
$}
\end{definition}
The second rule creates \emph{delay} steps of $t$ time units if no edge is enabled from now until just before $t$ time units have elapsed (third premise) but then, after exactly $t$ time units, some edge becomes enabled (second premise).
The first rule applies if an edge $\ell \xtr{G, a, R}_\Edges \mu$ is enabled:
a transition is taken with the edge's label, the successor state's location is chosen by $\mu$, $v$ is updated by resetting the clocks in $R$ to zero, and the expiration times for the restarted clocks are resampled.
All other expiration times remain unchanged.
Notice that $\sem{M}$ is also a nondeterministic labelled Markov processes~\cite{Wol12} (a proof can be found in~\cite{DArLM16}).

\begin{example}
\label{ex:SASemantics}
\Cref{fig:SASemanticsExample} outlines the semantics of $M_0$.
The first step from $\ell_0$ to all the states in $\ell_1$ is a single transition.
Its probability measure is the product of $F(x)$ and $F(y)$, sampling the expiration times of the two clocks.
We exemplify the behaviour of all of these states by showing it for the case of expiration times $e(x)$ and $e(y)$, with $e(x) < e(y)$.
In this case, to maximise the probability of reaching \cmark, we should take the transition to the state in $\ell_2$.
If a scheduler $\sched$ can see the expiration times, noting that only their order matters here, it can always make the optimal choice and achieve $\Pmax{\set{\text{\cmark}}}{\{\sched\}} = 1$.
\end{example}

\section{Classes of Schedulers}
\label{sec:Classes}

We now define classes of schedulers with restricted information, hiding in various combinations the history of states and parts of states, such as clock values and expiration times.
All definitions consider TPTS as in Definition~\ref{def:SASemantics}, and we require for all \sched that $\tuple{a, \mu} \in \support{\sched(h)} \Rightarrow \mathit{lst}_h \xtr{a}_\Trans \mu$, as in Definition~\ref{def:Scheduler}.

\subsection{Classic Schedulers}
\label{sec:ClassicSchedulers}

We first consider the ``classic'' complete-information setting where schedulers can in particular see expiration times.
We start with restricted classes of history-dependent schedulers.
Our first restriction hides the values of all clocks, only revealing the total time since the start of the run.
This is inspired by the step-counting or time-tracking schedulers needed to obtain optimal step-bounded or time-bounded reachability probabilities on MDP or Markov automata:

\begin{definition}
\label{def:ClassicHistLTE}
A classic history-dependent \emph{global-time} scheduler is a measurable function
$\sched\isfun(\States|_{\ell,t,e} \times \Alphabet \uplus \RRplus)^* \times \States|_{\ell,t,e} \to \Dist{\Alphabet \uplus \RRplus \times \Prob{\States}}$,
where $\States|_{\ell,t,e} \defeq \Loc \times \RRpluszero \times \Val$ with the second component being the total time $t$ elapsed since the start of the run. 
We write $\Sched^\mathit{hist}_{\ell,t,e}$ for the set of all such schedulers.
\end{definition}
We next hide the values of all clocks, revealing only their expiration times:

\begin{definition}
\label{def:ClassicHistLE}
A classic history-dependent \emph{location-based} scheduler is a measurable function
$\sched\isfun(\States|_{\ell,e} \times \Alphabet \uplus \RRplus)^* \times \States|_{\ell,e} \to \Dist{\Alphabet \uplus \RRplus \times \Prob{\States}}$,
where $\States|_{\ell,e} \defeq \Loc \times \Val$, with the second component being the clock expiration times~$e$.
$\Sched^\mathit{hist}_{\ell,e}$ is the set of all such schedulers.
\end{definition}
Having defined three classes of classic history-dependent schedulers, $\Sched^\mathit{hist}_{\ell,v,e}$, $\Sched^\mathit{hist}_{\ell,t,e}$ and $\Sched^\mathit{hist}_{\ell,e}$, noting that $\Sched^\mathit{hist}_{\ell,v,e}$ denotes all schedulers of Def.~\ref{def:Scheduler}, we also consider them with the restriction that they only see the relative order of clock expiration, instead of the exact expiration times:
for each pair of clocks $c_1,c_2$, these schedulers see the relation $\sim\:\in\{<,=,>\}$ in $e(c_1) - v(c_1) \sim e(c_2) - v(c_2)$.
E.g.\ in $\ell_1$ of Example~\ref{ex:SASemantics}, the scheduler would not see $e(x)$ and $e(y)$, but only whether $e(x) < e(y)$ or vice-versa (since $v(x) = v(y) = 0$, and equality has probability~0).
We consider this case because the expiration order is sufficient for the first algorithm of Bryans \etal\cite{BBD03}, and would allow optimal decisions in $M_0$ of \Cref{fig:SAExample}.
We denote the relative order information by $o$, and the corresponding scheduler classes by $\Sched^\mathit{hist}_{\ell,v,o}$, $\Sched^\mathit{hist}_{\ell,t,o}$ and $\Sched^\mathit{hist}_{\ell,o}$.
We now define memoryless schedulers, which only see the current state and are at the core of \eg MDP model checking.
On most formalisms, they suffice to obtain optimal reachability probabilities.

\begin{definition}
A classic \emph{memoryless} scheduler 
is a measurable function
$\sched\isfun\States \to \Dist{\Alphabet \uplus \RRplus \times \Prob{\States}}$.
$\Sched^\mathit{ml}_{\ell,v,e}$ is the set of all such schedulers.
\end{definition}
We apply the same restrictions as for history-dependent schedulers:

\begin{definition}
A classic memoryless global-time scheduler is a measurable function
$\sched\isfun\States|_{\ell,t,e} \to \Dist{\Alphabet \uplus \RRplus \times \Prob{\States}}$,
with $\States|_{\ell,t,e}$ as in Definition~\ref{def:ClassicHistLTE}.
We write $\Sched^\mathit{ml}_{\ell,t,e}$ for the set of all such schedulers.
\end{definition}

\begin{definition}
A classic memoryless location-based scheduler is a measurable function
$\sched\isfun\States|_{\ell,e} \to \Dist{\Alphabet \uplus \RRplus \times \Prob{\States}}$,
with $\States|_{\ell,e}$ as in Definition~\ref{def:ClassicHistLE}.
We write $\Sched^\mathit{ml}_{\ell,e}$ for the set of all such schedulers.
\end{definition}
Again, we also consider memoryless schedulers that only see the expiration order, so we have memoryless scheduler classes $\Sched^\mathit{ml}_{\ell,v,e}$, $\Sched^\mathit{ml}_{\ell,t,e}$, $\Sched^\mathit{ml}_{\ell,e}$, $\Sched^\mathit{ml}_{\ell,v,o}$, $\Sched^\mathit{ml}_{\ell,t,o}$ and $\Sched^\mathit{ml}_{\ell,o}$.
Class $\Sched^\mathit{ml}_{\ell,o}$ is particularly attractive because it has a compact finite domain.

\subsection{Non-Prophetic Schedulers}
\label{sec:NonPropheticSchedulers}

Consider the SA $M_0$ in \Cref{fig:SAExample}.
No matter which of the previously defined scheduler classes we choose, we always find a scheduler that achieves probability $1$ to reach \cmark, and a scheduler that achieves probability $0$.
This is because they can all see the expiration times or expiration order of $x$ and $y$ when in $\ell_1$.
When in $\ell_1$, $x$ and $y$ have not yet expired---this will only happen later, in $\ell_2$ or $\ell_3$---yet the schedulers already know which clock will ``win''.
The classic schedulers can thus be seen to make decisions based on the timing of \emph{future} events.
This \emph{prophetic} scheduling has already been observed in~\cite{BD04}, where a ``fix'' in the form of 
the \emph{spent lifetimes} semantics was proposed.
Hartmanns \etal\cite{HHK16} have shown that this not only still permits prophetic scheduling, but even admits \emph{divine} scheduling, where a scheduler can 
\emph{change} the future.
The authors propose a complex \emph{non-prophetic} semantics that provably removes all prophetic and divine behaviour.

Much of the complication of the non-prophetic semantics of~\cite{HHK16} is due to it being specified for open SA that include delayable actions.
For the closed SA setting of this paper, prophetic scheduling can be more easily excluded by hiding from the schedulers all information about what will happen in the future of the system's evolution.
This information is only contained in the expiration times $e$ or the expiration order $o$.
We can thus keep the 
semantics of \Cref{sec:SASemantics}{} 
and modify the definition of schedulers to exclude prophetic behaviour by construction.

In what follows, we thus also consider all scheduler classes of \Cref{sec:ClassicSchedulers} with the added constraint that the expiration times, resp.\ the expiration order, are not visible, resulting in the \emph{non-prophetic} classes $\Sched^\mathit{hist}_{\ell,v}$, $\Sched^\mathit{hist}_{\ell,t}$, $\Sched^\mathit{hist}_{\ell}$, $\Sched^\mathit{ml}_{\ell,v}$, $\Sched^\mathit{ml}_{\ell,t}$ and $\Sched^\mathit{ml}_{\ell}$.
Any non-prophetic scheduler can only reach \cmark{} of $M_0$ with probability~$\frac{1}{2}$.

\section{The Power of Schedulers}
\label{sec:PowerOfSchedulers}

Now that we have defined a number of classes of schedulers, 
we need to determine what the effect of the restrictions is on our ability to optimally control an SA.
We thus evaluate the power of scheduler classes \wrt unbounded reachability probabilities (\Cref{def:ReachProb}) on the semantics of SA.
We will see that this simple setting already suffices to reveal interesting differences between scheduler classes.

For two scheduler classes $\Sched_1$ and $\Sched_2$, we write $\Sched_1 \schedgeq \Sched_2$ if, for all SA and all sets of goal locations $G$, $\Pmin{G}{\Sched_1} \leq \Pmin{G}{\Sched_2}$ and $\Pmax{G}{\Sched_1} \geq \Pmax{G}{\Sched_2}$.
We write $\Sched_1 \schedgtr \Sched_2$ if additionally there exists at least one SA and set $G'$ where $\Pmin{G'}{\Sched_1} < \Pmin{G'}{\Sched_2}$ or $\Pmax{G'}{\Sched_1} > \Pmax{G'}{\Sched_2}$.
Finally, we write $\Sched_1 \schedeq \Sched_2$ for $\Sched_1 \schedgeq \Sched_2 \wedge \Sched_2 \schedgeq \Sched_1$, and $\Sched_1 \schedneq \Sched_2$, \ie the classes are incomparable, for $\Sched_1 \schedngeq \Sched_2 \wedge \Sched_2 \schedngeq \Sched_1$.
Unless noted otherwise, we omit proofs for $\Sched_1 \schedgeq \Sched_2$ when it is obvious that the information available to $\Sched_1$ includes the information available to $\Sched_2$.
Our proofs are all based on the resolution of a single nondeterministic choice between two actions, to eventually reach one of two locations.
We therefore prove only \wrt the maximum probability, $p_{\max}$, since in all cases the minimum probability is given by $1-p_{\max}$ and an analogous proof for $p_{\min}$ can be made by relabelling locations.
In several proofs that use distinguishing SA with a similar structure to $M_0$, we write $\Pmax{\Sched_x^y}{}$ for $\Pmax{\set{\text{\cmark}}}{\Sched_x^y}$ to improve readability.

\subsection{The Classic Hierarchy}
\label{sec:PropheticHierarchy}

We first establish that all classic history-dependent scheduler classes are equivalent:

\begin{figure}[t]
\begin{minipage}[b]{0.6\textwidth}
\centering
\begin{tikzpicture}[on grid,x=1.75cm,y=1.3cm]
  \tikzstyle{every node}=[font=\normalsize]
  \node (histlve) {$\Sched^\mathit{hist}_{\ell,v,e}$};
  \node (histlte) [left=1 of histlve] {$\Sched^\mathit{hist}_{\ell,t,e}$};
  \node (histle) [left=1 of histlte] {$\Sched^\mathit{hist}_{\ell,e}$};
  \node (mllve) [above=1 of histlve] {$\Sched^\mathit{ml}_{\ell,v,e}$};
  \node (mllte) [left=1 of mllve] {$\Sched^\mathit{ml}_{\ell,t,e}$};
  \node (mlle) [left=1 of mllte] {$\Sched^\mathit{ml}_{\ell,e}$};
  \node (mllvo) [above=1 of mllve] {$\Sched^\mathit{ml}_{\ell,v,o}$};
  \node (mllto) [left=1 of mllvo] {$\Sched^\mathit{ml}_{\ell,t,o}$};
  \node (mllo) [left=1 of mllto] {$\Sched^\mathit{ml}_{\ell,o}$};
  \node (histlvo) [right=1 of histlve] {$\Sched^\mathit{hist}_{\ell,v,o}$};
  \node (histlto) [above=1 of histlvo] {$\Sched^\mathit{hist}_{\ell,t,o}$};
  \node (histlo) [above=1 of histlto] {$\Sched^\mathit{hist}_{\ell,o}$};
  \path[sloped] (histlve) -- node[midway] {\schedeq} (histlte);
  \path[sloped] (histlte) -- node[midway] {\schedeq} (histle);
  \path[sloped] (histlve) -- node[midway] {\schedeq} (mllve);
  \path[sloped] (histlte) -- node[midway] {\schedgtr} (mllte);
  \path[sloped] (mllve) -- node[midway] {\schedlss} (mllte);
  \path[sloped] (mllte) -- node[midway] {\schedlss} (mlle);
  \path[sloped] (histle) -- node[midway] {\schedgtr} (mlle);
  \path[sloped] (mllve) -- node[midway] {\schedgtr} (mllvo);
  \path[sloped] (mllte) -- node[midway] {\schedneq} (mllto);
  \path[sloped] (mlle) -- node[midway] {\schedneq} (mllo);
  \path[sloped] (mllvo) -- node[midway] {\schedlss} (mllto);
  \path[sloped] (mllto) -- node[midway] {\schedlss} (mllo);
  \path[sloped] (histlve) -- node[midway] {\schedgtr} (histlvo);
  \path[sloped] (histlvo) -- node[midway] {\schedeq} (histlto);
  \path[sloped] (histlto) -- node[midway] {\schedeq} (histlo);
  \path[sloped] (histlto) -- node[midway] {\schedgtr} (mllve);
  \path[sloped] (histlo) -- node[midway] {\schedlss} (mllvo);
  ;
\end{tikzpicture}
\caption{Hierarchy of classic scheduler classes}
\label{fig:PropheticHierarchy}
\end{minipage}%
\begin{minipage}[b]{0.4\textwidth}
\centering
\begin{tikzpicture}[on grid,x=1.55cm,y=1.3cm]
  \tikzstyle{every node}=[font=\normalsize]
  \node (histlv) {$\Sched^\mathit{hist}_{\ell,v}$};
  \node (histlt) [left=1 of histlv] {$\Sched^\mathit{hist}_{\ell,t}$};
  \node (histl) [left=1 of histlt] {$\Sched^\mathit{hist}_{\ell}$};
  \node (mllv) [above=1 of histlv] {$\Sched^\mathit{ml}_{\ell,v}$};
  \node (mllt) [left=1 of mllv] {$\Sched^\mathit{ml}_{\ell,t}$};
  \node (mll) [left=1 of mllt] {$\Sched^\mathit{ml}_{\ell}$};
  \path[sloped] (histlv) -- node[midway] {\schedeq} (histlt);
  \path[sloped] (histlt) -- node[midway] {\schedeq} (histl);
  \path[sloped] (histlv) -- node[midway] {\schedgtr} (mllv);
  \path[sloped] (histlt) -- node[midway] {\schedgtr} (mllt);
  \path[sloped] (mllv) -- node[midway] {\schedlss} (mllt);
  \path[sloped] (mllt) -- node[midway] {\schedlss} (mll);
  \path[sloped] (histl) -- node[midway] {\schedgtr} (mll);
  ;
\end{tikzpicture}
\caption{Non-prophetic classes}
\label{fig:NonPropheticHierarchy}
\end{minipage}
\end{figure}

\begin{proposition}
\label{prop:EquivalenceClassicHist}
$\Sched^\mathit{hist}_{\ell,v,e} \schedeq \Sched^\mathit{hist}_{\ell,t,e} \schedeq \Sched^\mathit{hist}_{\ell,e}$.
\end{proposition}
\begin{proof}\proofnegspace{}
From the transition labels in $\Alphabet \uplus \RRplus$ in the history $(\States' \times \Alphabet \uplus \RRplus)^*$, with $\States' \in \set{ \States, \States|_{\ell,t,e}, \States|_{\ell,e} }$ depending on the scheduler class, we can reconstruct the total elapsed time as well as the values of all clocks:
to obtain the total elapsed time, sum the labels in $\RRplus$ up to each state;
to obtain the values of all clocks, do the same per clock and perform the \restart{s} of the edges identified by the~actions.
\end{proof}
The same argument applies among the expiration-order history-dependent classes:

\begin{proposition}
\label{prop:EquivalenceOrderedHist}
$\Sched^\mathit{hist}_{\ell,v,o} \schedeq \Sched^\mathit{hist}_{\ell,t,o} \schedeq \Sched^\mathit{hist}_{\ell,o}$.
\end{proposition}
However, the expiration-order history-dependent schedulers are strictly less powerful than the classic history-dependent ones:

\begin{proposition}
\label{prop:IncompClassicOrderedHist}
$\Sched^\mathit{hist}_{\ell,v,e} \schedgtr \Sched^\mathit{hist}_{\ell,v,o}$.
\end{proposition}
\begin{proof}\proofnegspace{}
Consider the SA~$M_1$ in \Cref{fig:Cex2}.
Note that the history does not provide any information for making the choice in $\ell_1$: we always arrive after having spent zero time in $\ell_0$ and then having taken the single edge to $\ell_1$.
We can analytically determine that $\Pmax{\Sched^\mathit{hist}_{\ell,v,e}}{} = \frac{3}{4}$
by going from $\ell_1$ to $\ell_2$ if $e(x) \leq \frac{1}{2}$ and to $\ell_3$ otherwise.
We would obtain a probability equal to $\frac{1}{2}$ by always going to either $\ell_2$ or $\ell_3$ or by picking either edge with equal probability.
This is the best we can do if $e$ is not visible, and thus $\Pmax{\Sched^\mathit{hist}_{\ell,v,o}}{} = \frac{1}{2}$:
in $\ell_1$, $v(x) = v(y) = 0$ and the expiration order is always ``$y$ before $x$'' because $y$ has not yet been started.
\end{proof}
Just like for MDP and unbounded reachability probabilities, the classic history-dependent and memoryless schedulers with complete information are equivalent:

\begin{proposition}
\label{prop:ClassicHistMemorylessEquiv}
$\Sched^\mathit{hist}_{\ell,v,e} \schedeq \Sched^\mathit{ml}_{\ell,v,e}$.
\end{proposition}
\begingroup
\renewcommand*{\proofname}{Proof sketch}
\begin{proof}\proofnegspace{}
Our definition of TPTS only allows finite nondeterministic choices, \ie we have a very restricted form of continuous-space MDP.
We can thus adapt the argument of the corresponding proof for MDP~\cite[Lemma 10.102]{BK08}:
For each state (of possibly countably many), we construct a notional optimal memoryless (and deterministic) scheduler in the same way, replacing the summation by an integration for the continuous measures in the transition function.
It remains to show that this scheduler is indeed measurable.
For TPTS that are the semantics of SA, this follows from the way clock values are used in the guard sets so that optimal decisions are constant over intervals of clock values and expiration times (see \eg the arguments in~\cite{BBD03} or \cite{KNSS00}).
\end{proof}
\endgroup\noindent
On the other hand, when restricting schedulers to see the expiration order only, history-dependent and memoryless schedulers are no longer equivalent:

\begin{figure}[t]
\begin{minipage}[b]{0.33\textwidth}
\centering
\begin{tikzpicture}[on grid,auto]
  \node[state] (l0) {$\ell_0$};
  \coordinate[left=0.3 of l0.west] (start);
  \node[] (me) [above left=0.4 and 1.1 of l0] {\small$M_1$:};
  \node[] (distr) [above right=0.2 and 1.3 of l0,align=left] {$x\colon \textsc{Uni}(0, 1)$\\$y\colon \textsc{Uni}(0, 1)$};
  \node[state] (l1) [below=1 of l0] {$\ell_1$};
  \node[state] (l2) [below left=0.875 and 0.75 of l1] {$\ell_2$};
  \node[state] (l3) [below right=0.875 and 0.75 of l1] {$\ell_3$};
  \node[state] (yes) [below=1 of l2] {\cmark};
  \node[state] (no) [below=1 of l3] {\xmark};
  ;
  \path[->]
    (start) edge node {} (l0)
    (l0) edge node[right,pos=0.25,inner sep=0.5mm] {\strut$\varnothing$} node[right,pos=0.7,inner sep=0.5mm] {\strut$\text{\restart}(\{ x \})$} (l1)
    (l1) edge[] node[left,pos=0.15,inner sep=1mm] {\strut$\varnothing\!$} node[left,pos=0.55,inner sep=1mm] {\strut$\text{\restart}(\{ y \})$} (l2)
    (l1) edge[] node[right,pos=0.15,inner sep=1mm] {\strut$\varnothing$} node[right,pos=0.55,inner sep=1mm] {\strut$\text{\restart}(\{ y \})$} (l3)
    (l2) edge node[left,pos=0.25,inner sep=0.5mm] {\strut$\{ x \}$} (yes)
    (l2) edge[bend right=20] node[above,pos=0.33,inner sep=2mm] {\strut$\{ y \}$} (no)
    (l3) edge[bend left=20] node[above,pos=0.33,inner sep=2mm] {\strut$\{ y \}$} (yes)
    (l3) edge node[right,pos=0.25,inner sep=0.5mm] {\strut$\{ x \}$} (no)
  ;
\end{tikzpicture}
\caption{SA $M_1$}
\label{fig:Cex2}
\end{minipage}%
\begin{minipage}[b]{0.33\textwidth}
\centering
\begin{tikzpicture}[on grid,auto]
  \node[state] (l0) {$\ell_0$};
  \coordinate[left=0.3 of l0.west] (start);
  \node[] (me) [above left=0.4 and 1.1 of l0] {\small$M_2$:};
  \node[] (distr) [above right=0.2 and 1.3 of l0,align=left] {$x\colon \textsc{Uni}(0, 8)$\\$y\colon \textsc{Uni}(0, 1)$\\$z\colon \textsc{Uni}(0, 4)$};
  \node[state] (l1) [below=1 of l0] {$\ell_1$};
  \node[state] (l2) [below=1 of l1] {$\ell_2$};
  \node[state] (l3) [below left=0.875 and 0.75 of l2] {$\ell_3$};
  \node[state] (l4) [below right=0.875 and 0.75 of l2] {$\ell_4$};
  \node[state] (yes) [below=1 of l3] {\cmark};
  \node[state] (no) [below=1 of l4] {\xmark};
  ;
  \path[->]
    (start) edge node {} (l0)
    (l0) edge node[right,pos=0.25,inner sep=0.5mm] {\strut$\varnothing$} node[right,pos=0.7,inner sep=0.5mm] {\strut$\text{\restart}(\{ x, z \})$} (l1)
    (l1) edge[bend right=50] node[left,pos=0.25,inner sep=0.5mm] {\strut$\{ x \}$} node[left,pos=0.7,inner sep=0.5mm] {\strut$\text{\restart}(\{z \})$} (l2)
    (l1) edge[bend left=50] node[right,pos=0.25,inner sep=0.5mm] {\strut$\{z \}$} node[right,pos=0.7,inner sep=0.75mm] {\strut$\text{\restart}(\{z \})$} (l2)
    (l2) edge[] node[left,pos=0.15,inner sep=1mm] {\strut$\varnothing\!$} node[left,pos=0.55,inner sep=1mm] {\strut$\text{\restart}(\{ y \})$} (l3)
    (l2) edge[] node[right,pos=0.15,inner sep=1mm] {\strut$\varnothing$} node[right,pos=0.55,inner sep=1mm] {\strut$\text{\restart}(\{ y \})$} (l4)
    (l3) edge node[left,pos=0.25,inner sep=0.5mm] {\strut$\{ x \}$} (yes)
    (l3) edge[bend right=20] node[above,pos=0.33,inner sep=2mm] {\strut$\{ y \}$} (no)
    (l4) edge[bend left=20] node[above,pos=0.33,inner sep=2mm] {\strut$\{ y \}$} (yes)
    (l4) edge node[right,pos=0.25,inner sep=0.5mm] {\strut$\{ x \}$} (no)
  ;
\end{tikzpicture}
\caption{SA $M_2$}
\label{fig:CexX}
\end{minipage}%
\begin{minipage}[b]{0.33\textwidth}
\centering
\begin{tikzpicture}[on grid,auto]
  \node[state] (l0) {$\ell_0$};
  \coordinate[left=0.3 of l0.west] (start);
  \node[] (me) [above left=0.4 and 1.1 of l0] {\small$M_3$:};
  \node[] (distr) [above right=0.2 and 1.3 of l0,align=left] {$x\colon \textsc{Uni}(0, 1)$\\$y\colon \textsc{Uni}(0, 1)$\\$z\colon \textsc{Uni}(0, 1)$};
  \node[state] (l1) [below=1 of l0] {$\ell_1$};
  \node[state] (l2) [below=1 of l1] {$\ell_2$};
  \node[state] (l3) [below=1 of l2] {$\ell_3$};
  \node[state] (l4) [below left=0.875 and 0.75 of l3] {$\ell_4$};
  \node[state] (l5) [below right=0.875 and 0.75 of l3] {$\ell_5$};
  \node[state] (yes) [below=1 of l4] {\cmark};
  \node[state] (no) [below=1 of l5] {\xmark};
  ;
  \path[->]
    (start) edge node {} (l0)
    (l0) edge node[right,pos=0.25,inner sep=0.5mm] {\strut$\varnothing$} node[right,pos=0.7,inner sep=0.5mm] {\strut$\text{\restart}(\{ z \})$} (l1)
    (l1) edge node[right,pos=0.25,inner sep=0.5mm] {\strut$\{ z \}$} node[right,pos=0.7,inner sep=0.5mm] {\strut$\text{\restart}(\{ x, y, z \})$} (l2)
    (l2) edge[bend right=50] node[left,pos=0.25,inner sep=0.5mm] {\strut$\{ x \}$} node[left,pos=0.7,inner sep=0.5mm] {\strut$\text{\restart}(\{ y \})$} (l3)
    (l2) edge[bend left=50] node[right,pos=0.25,inner sep=0.5mm] {\strut$\{ y \}$} node[right,pos=0.7,inner sep=0.75mm] {\strut$\text{\restart}(\{ x \})$} (l3)
    (l3) edge[] node[left,pos=0.15,inner sep=1mm] {\strut$\varnothing\!$} (l4)
    (l3) edge[] node[right,pos=0.15,inner sep=1mm] {\strut$\varnothing$} (l5)
    (l4) edge node[left,pos=0.25,inner sep=0.5mm] {\strut$\{ x \}$} (yes)
    (l4) edge[bend right=20] node[above,pos=0.33,inner sep=2mm] {\strut$\{ y \}$} (no)
    (l5) edge[bend left=20] node[above,pos=0.33,inner sep=2mm] {\strut$\{ y \}$} (yes)
    (l5) edge node[right,pos=0.25,inner sep=0.5mm] {\strut$\{ x \}$} (no)
  ;
\end{tikzpicture}
\caption{SA $M_3$}
\label{fig:Cex3}
\end{minipage}%
\end{figure}

\begingroup
\begin{proposition}
\label{prop:OrderedHistVsMemoryless}
$\Sched^\mathit{hist}_{\ell,v,o} \schedgtr \Sched^\mathit{ml}_{\ell,v,o}$.
\end{proposition}
\begin{proof}\proofnegspace{}
Consider the SA~$M_2$ in \Cref{fig:CexX}.
Let $\sched^{\mathit{opt}}_{\mathit{ml}(l,v,o)}$ be the (unknown) optimal scheduler in $\Sched^\mathit{ml}_{\ell,v,o}$ \wrt the max.\ probability of reaching \cmark.
Define $\sched^{\mathit{better}}_{\mathit{hist}(l,v,o)} \in \Sched^\mathit{hist}_{\ell,v,o}$ as:
When in $\ell_2$ and the last edge in the history is the left one (\ie $x$ is expired), go to $\ell_3$; otherwise, behave like $\sched^{\mathit{opt}}_{\mathit{ml}(l,v,o)}$.
This scheduler distinguishes $\Sched^\mathit{hist}_{\ell,v,o}$ and $\Sched^\mathit{ml}_{\ell,v,o}$ (by achieving a strictly higher max.\ probability than $\sched^{\mathit{opt}}_{\mathit{ml}(l,v,o)}$) if and only if there are some combinations of clock values (aspect $v$) and expiration orders (aspect $o$) in $\ell_2$ that can be reached with positive probability via the left edge into $\ell_2$, for which $\sched^{\mathit{opt}}_{\mathit{ml}(l,v,o)}$ must nevertheless decide to go to $\ell_4$.

All possible clock valuations in $\ell_2$ can be achieved via either the left or right edges, but taking the left edge implies that $x$ expires before $z$ in $\ell_2$.
It is thus sufficient to show that $\sched^{\mathit{opt}}_{\mathit{ml}(l,v,o)}$ must go to $\ell_4$ in some cases where $x$ expires before $z$.
Now consider 
the scheduler that goes to $\ell_3$ whenever $x$ expires before $z$.
Under this scheduler, 
the max.\ probability is $\frac{77}{96} = 0.80208\bar3$.
Since this is the only scheduler in $\Sched^\mathit{ml}_{\ell,v,o}$ that never goes to $l_4$ when $x$ expires before $z$, it remains to show that the max.\ probability under $\sched^{\mathit{opt}}_{\mathit{ml}(l,v,o)}$ is greater than $\frac{77}{96}$.
Now consider the scheduler that goes to $\ell_3$ whenever $(x\text{ expires before }z) \wedge v(x) < \frac{35}{12}$.
Under this scheduler, we have a max.\ probability of $\frac{7561}{9216} \approx 0.820421$.
Thus $\sched^{\mathit{opt}}_{\mathit{ml}(l,v,o)}$ must sometimes go to $l_4$ even when the left edge was taken, so $\sched^{\mathit{better}}_{\mathit{hist}(l,v,o)}$ achieves a higher max.\ probability and thus distinguishes the scheduler classes.
\end{proof}
\endgroup\noindent
Knowing only the global elapsed time is less powerful than knowing the full history or the values of all clocks:

\begin{proposition}
\label{prop:ClassicValVsTime}
$\Sched^\mathit{hist}_{\ell,t,e} \schedgtr \Sched^\mathit{ml}_{\ell,t,e}$ and $\Sched^\mathit{ml}_{\ell,v,e} \schedgtr \Sched^\mathit{ml}_{\ell,t,e}$.
\end{proposition}
\begingroup
\renewcommand*{\proofname}{Proof sketch}
\begin{proof}\proofnegspace{}
Consider the SA~$M_3$ in \Cref{fig:Cex3}.
We have $\Pmax{\Sched^\mathit{hist}_{\ell,t,e}}{} = 1$:
when in $\ell_3$, the scheduler sees from the history which of the two incoming edges was used, and thus knows whether $x$ or $y$ is already expired.
It can then make the optimal choice: go to $\ell_4$ if $x$ is already expired, or to $\ell_5$ otherwise (if $y$ is already expired).
We also have $\Pmax{\Sched^\mathit{ml}_{\ell,v,e}}{} = 1$:
the scheduler sees that either $v(x) = 0$ or $v(y) = 0$, which implies that the other clock is already expired, and the argument above applies.
However, $\Pmax{\Sched^\mathit{ml}_{\ell,t,e}}{} < 1$:
the distribution of elapsed time $t$ on entering $\ell_3$ is independent of which edge is taken, so $t$ reveals no information about whether $x$ or $y$ has already expired.
Since $v$ is not visible, the expiration times in $e$ are not useful: they are both positive and drawn from the same distribution, but one unknown clock is expired.
The wait for $z$ in $\ell_1$ ensures that comparing $t$ with the expiration times in $e$ does not reveal any information.
\end{proof}
\endgroup
In the case of MDP, knowing the total elapsed time (\ie steps) does not make a difference for unbounded reachability.
Only for step-bounded properties is that extra knowledge necessary to achieve optimal probabilities.
With SA, however, it makes a difference even in the unbounded case: 

\begin{figure}[t]
\begin{minipage}[b]{0.33\textwidth}
\centering
\begin{tikzpicture}[on grid,auto]
  \node[state] (l0) {$\ell_0$};
  \coordinate[left=0.3 of l0.west] (start);
  \node[] (me) [above left=0.4 and 1.0 of l0] {\small$M_4$:};
  \node[] (distr) [above right=0.2 and 1.3 of l0,align=left] {$x\colon \textsc{Uni}(0, 2)$\\$y\colon \textsc{Uni}(0, 1)$\\$z\colon \textsc{Uni}(0, 2)$};
  \node[state] (l1) [below=1 of l0] {$\ell_1$};
  \node[state] (l2) [below=1 of l1] {$\ell_2$};
  \node[state] (l3) [below left=0.875 and 0.75 of l2] {$\ell_3$};
  \node[state] (l4) [below right=0.875 and 0.75 of l2] {$\ell_4$};
  \node[state] (yes) [below=1 of l3] {\cmark};
  \node[state] (no) [below=1 of l4] {\xmark};
  ;
  \path[->]
    (start) edge node {} (l0)
    (l0) edge node[right,pos=0.25,inner sep=0.5mm] {\strut$\varnothing$} node[right,pos=0.7,inner sep=0.5mm] {\strut$\text{\restart}(\{ x, z \})$} (l1)
    (l1) edge node[right,pos=0.25,inner sep=0.5mm] {\strut$\{ z \}$} node[right,pos=0.7,inner sep=0.5mm] {\strut$\text{\restart}(\{ y, z \})$} (l2)
    (l2) edge[] node[left,pos=0.15,inner sep=1mm] {\strut$\varnothing\!$} (l3)
    (l2) edge[] node[right,pos=0.15,inner sep=1mm] {\strut$\varnothing$} (l4)
    (l3) edge node[left,pos=0.25,inner sep=0.5mm] {\strut$\{ x \}$} (yes)
    (l3) edge[bend right=20] node[above,pos=0.33,inner sep=2mm] {\strut$\{ y \}$} (no)
    (l4) edge[bend left=20] node[above,pos=0.33,inner sep=2mm] {\strut$\{ y \}$} (yes)
    (l4) edge node[right,pos=0.25,inner sep=0.5mm] {\strut$\{ x \}$} (no)
  ;
\end{tikzpicture}
\caption{SA $M_4$}
\label{fig:Cex4}
\end{minipage}%
\begin{minipage}[b]{0.33\textwidth}
\centering
\begin{tikzpicture}[on grid,auto]
  \node[state] (l0) {$\ell_1$};
  \coordinate[left=0.3 of l0.west] (start);
  \node[] (me) [above left=0.4 and 1.1 of l0] {\small$M_5$:};
  \node[] (distr) [above right=0.2 and 1.3 of l0,align=left] {$x\colon \textsc{Uni}(0, 1)$\\$y\colon \textsc{Uni}(0, 1)$};
  \node[state] (l1) [below=1 of l0] {$\ell_2$};
  \node[state] (l2) [below=1 of l1] {$\ell_3$};
  \node[state] (l3) [below left=0.875 and 0.75 of l2] {$\ell_4$};
  \node[state] (l4) [below right=0.875 and 0.75 of l2] {$\ell_5$};
  \node[state] (yes) [below=1 of l3] {\cmark};
  \node[state] (no) [below=1 of l4] {\xmark};
  ;
  \path[->]
    (start) edge node {} (l0)
    (l0) edge node[right,pos=0.25,inner sep=0.5mm] {\strut$\varnothing$} node[right,pos=0.7,inner sep=0.5mm] {\strut$\text{\restart}(\{ x, y \})$} (l1)
    (l1) edge node[right,pos=0.25,inner sep=0.5mm] {\strut$\{ y \}$} node[right,pos=0.7,inner sep=0.5mm] {\strut$\text{\restart}(\{ y \})$} (l2)
    (l2) edge[] node[left,pos=0.15,inner sep=1mm] {\strut$\varnothing\!$} node[left,pos=0.55,inner sep=1mm] {\strut$\text{\restart}(\{ y \})$} (l3)
    (l2) edge[] node[right,pos=0.15,inner sep=1mm] {\strut$\varnothing$} node[right,pos=0.55,inner sep=1mm] {\strut$\text{\restart}(\{ y \})$} (l4)
    (l3) edge node[left,pos=0.25,inner sep=0.5mm] {\strut$\{ x \}$} (yes)
    (l3) edge[bend right=20] node[above,pos=0.33,inner sep=2mm] {\strut$\{ y \}$} (no)
    (l4) edge[bend left=20] node[above,pos=0.33,inner sep=2mm] {\strut$\{ y \}$} (yes)
    (l4) edge node[right,pos=0.25,inner sep=0.5mm] {\strut$\{ x \}$} (no)
  ;
\end{tikzpicture}
\caption{SA $M_5$}
\label{fig:CexY}
\end{minipage}%
\begin{minipage}[b]{0.33\textwidth}
\centering
\begin{tikzpicture}[on grid,auto]
  \node[state] (l0) {$\ell_0$};
  \coordinate[left=0.3 of l0.west] (start);
  \node[] (me) [above left=0.4 and 1.0 of l0] {\small$M_6$:};
  \node[] (distr) [above right=0.2 and 1.3 of l0,align=left] {$x\colon \textsc{Uni}(0, 1)$\\$y\colon \textsc{Uni}(0, 1)$};
  \node[state] (l1) [below=1 of l0] {$\ell_1$};
  \node[state] (l2) [below=1 of l1] {$\ell_2$};
  \node[state] (l3) [below left=0.875 and 0.75 of l2] {$\ell_3$};
  \node[state] (l4) [below right=0.875 and 0.75 of l2] {$\ell_4$};
  \node[state] (yes) [below=1 of l3] {\cmark};
  \node[state] (no) [below=1 of l4] {\xmark};
  ;
  \path[->]
    (start) edge node {} (l0)
    (l0) edge node[right,pos=0.25,inner sep=0.5mm] {\strut$\varnothing$} node[right,pos=0.7,inner sep=0.5mm] {\strut$\text{\restart}(\{ x, y \})$} (l1)
    (l1) edge[bend right=50] node[left,pos=0.25,inner sep=0.5mm] {\strut$\{ x \}$} (l2)
    (l1) edge[bend left=50] node[right,pos=0.25,inner sep=0.5mm] {\strut$\{ y \}$} (l2)
    (l2) edge[] node[left,pos=0.15,inner sep=1mm] {\strut$\varnothing\!$} (l3)
    (l2) edge[] node[right,pos=0.15,inner sep=1mm] {\strut$\varnothing$} (l4)
    (l3) edge node[left,pos=0.25,inner sep=0.5mm] {\strut$\{ x \}$} (yes)
    (l3) edge[bend right=20] node[above,pos=0.33,inner sep=2mm] {\strut$\{ y \}$} (no)
    (l4) edge[bend left=20] node[above,pos=0.33,inner sep=2mm] {\strut$\{ y \}$} (yes)
    (l4) edge node[right,pos=0.25,inner sep=0.5mm] {\strut$\{ x \}$} (no)
  ;
\end{tikzpicture}
\caption{SA $M_6$}
\label{fig:Cex1}
\end{minipage}
\end{figure}

\begin{proposition}
\label{prop:ClassicTimeVersusPureML}
$\Sched^\mathit{ml}_{\ell,t,e} \schedgtr \Sched^\mathit{ml}_{\ell,e}$.
\end{proposition}
\begin{proof}\proofnegspace{}
Consider SA~$M_4$ in \Cref{fig:Cex4}.
We have $\Pmax{\Sched^\mathit{ml}_{\ell,t,e}}{} = 1$:
in $\ell_2$, the remaining time until $y$ expires is $e(y)$ and the remaining time until $x$ expires is $e(x) - t$ for the global time value $t$ as $\ell_2$ is entered.
The scheduler can observe all of these quantities and thus optimally go to $\ell_3$ if $x$ will expire first, or to $\ell_4$ otherwise.
However, $\Pmax{\Sched^\mathit{ml}_{\ell,e}}{} < 1$:
$e(x)$ only contains the absolute expiration time of $x$, but without knowing $t$ or the expiration time of $z$ in $\ell_1$, and thus the current value $v(x)$, this scheduler cannot know with certainty which of the clocks will expire first and is therefore unable to make an optimal choice in $\ell_2$.
\end{proof}
Finally, we need to compare the memoryless schedulers that see the clock expiration times with memoryless schedulers that see the expiration order.
As noted in \Cref{sec:ClassicSchedulers}, these two views of the current state are incomparable unless we also see the clock values:

\begin{proposition}
\label{prop:ClassicVsOrderVal}
$\Sched^\mathit{ml}_{\ell,v,e} \schedgtr \Sched^\mathit{ml}_{\ell,v,o}$.
\end{proposition}
\begin{proof}\proofnegspace{}
$\Sched^\mathit{ml}_{\ell,v,e} \schednleq \Sched^\mathit{ml}_{\ell,v,o}$ follows from the same argument as in the proof of \Cref{prop:IncompClassicOrderedHist}.
$\Sched^\mathit{ml}_{\ell,v,e} \schedgeq \Sched^\mathit{ml}_{\ell,v,o}$ is because knowing the current clock values $v$ and the expiration times $e$ is equivalent to knowing the expiration order, since that is precisely the order of the differences $e(c) - v(c)$ for all clocks~$c$.
\end{proof}

\begin{proposition}
\label{prop:ClassicVsOrderTime}
$\Sched^\mathit{ml}_{\ell,t,e} \schedneq \Sched^\mathit{ml}_{\ell,t,o}$.
\end{proposition}
\begin{proof}\proofnegspace{}
$\Sched^\mathit{ml}_{\ell,t,e} \schednleq \Sched^\mathit{ml}_{\ell,t,o}$ follows from the same argument as in the proof of \Cref{prop:IncompClassicOrderedHist}.
For $\Sched^\mathit{ml}_{\ell,t,e} \schedngeq \Sched^\mathit{ml}_{\ell,t,o}$, consider the SA~$M_3$ of \Cref{fig:Cex3}.
We know from the proof of \Cref{prop:ClassicValVsTime} that $\Pmax{\Sched^\mathit{ml}_{\ell,t,e}}{} < 1$.
However, if the scheduler knows the order in which the clocks will expire, it knows which one has already expired (the first one in the order), and can thus make the optimal choice in $\ell_3$ to achieve $\Pmax{\Sched^\mathit{ml}_{\ell,t,o}}{} = 1$ (again via the same arguments as in the proof of \Cref{prop:ClassicValVsTime}).
\end{proof}

\begin{proposition}
\label{prop:ClassicVsOrderSimple}
$\Sched^\mathit{ml}_{\ell,e} \schedneq \Sched^\mathit{ml}_{\ell,o}$.
\end{proposition}
\begin{proof}\proofnegspace{}
The argument of \Cref{prop:ClassicVsOrderTime} applies if we observe that, in SA~$M_3$ of \Cref{fig:Cex3}, we also have $\Pmax{\Sched^\mathit{ml}_{\ell,e}}{} < 1$ via the same argument for $\Sched^\mathit{ml}_{\ell,t,e}$ in the proof of \Cref{prop:ClassicValVsTime}.
\end{proof}
Among the expiration-order schedulers, the hierarchy is as expected:

\begingroup
\renewcommand*{\proofname}{Proof sketch}
\begin{proposition}
\label{prop:WithinOrder}
$\Sched^\mathit{ml}_{\ell,v,o} \schedgtr \Sched^\mathit{ml}_{\ell,t,o} \schedgtr \Sched^\mathit{ml}_{\ell,o}$.
\end{proposition}
\begin{proof}\proofnegspace{}
Consider the SA $M_5$ in \Cref{fig:CexY}.
To maximise the probability, in $\ell_3$ we should go to $\ell_4$ whenever $x$ is already expired or close to expiring, for which the amount of time spent in $\ell_2$ is an indicator.
$\Sched^\mathit{ml}_{\ell,o}$ only knows that $x$ may have expired when the expiration order is ``$x$ before $y$'', but definitely has not expired when it is ``$y$ before $x$''.
Schedulers in $\Sched^\mathit{ml}_{\ell,t,o}$ can do better:
They also see the exact amount of time spent in $\ell_2$.
Thus $\Sched^\mathit{ml}_{\ell,t,o} \schedgtr \Sched^\mathit{ml}_{\ell,o}$.
If we modify $M_5$ by adding an initial $e(z) \sim \textsc{Uni}(0, 1)$ delay from $\ell_0$ to $\ell_1$, as in $M_3$, then the same argument can be used to prove $\Sched^\mathit{ml}_{\ell,v,o} \schedgtr \Sched^\mathit{ml}_{\ell,t,o}$, because the extra delay makes knowing the global elapsed time $t$ useless, but it is visible to $\Sched^\mathit{ml}_{\ell,v,o}$ as $v(x)$.
\end{proof}
\endgroup\noindent
We have thus established the hierarchy of classic schedulers shown in \Cref{fig:PropheticHierarchy}, noting that some of the relationships follow from the propositions by transitivity.

\subsection{The Non-Prophetic Hierarchy}

Each non-prophetic scheduler class is clearly dominated by the classic and expira\-tion-order scheduler classes that otherwise have the same information, for example $\Sched^\mathit{hist}_{\ell,v,e} \schedgtr \Sched^\mathit{hist}_{\ell,v}$ (with very simple distinguishing SA).
We show that the non-prophetic hierarchy follows the shape of the classic case, including the difference between global-time and pure memoryless schedulers, with the notable exception of memoryless schedulers being weaker than history-dependent ones.


\begin{proposition}
\label{prop:EquivalenceNonProphHist}
$\Sched^\mathit{hist}_{\ell,v} \schedeq \Sched^\mathit{hist}_{\ell,t} \schedeq \Sched^\mathit{hist}_{\ell}$.
\end{proposition}
\begin{proof}\proofnegspace{}
This follows from the argument of \Cref{prop:EquivalenceClassicHist}.
\end{proof}

\begin{proposition}
\label{prop:NonProphHistVsFullMemoryless}
$\Sched^\mathit{hist}_{\ell,v} \schedgtr \Sched^\mathit{ml}_{\ell,v}$.
\end{proposition}
\begin{proof}\proofnegspace{}
Consider the SA $M_6$ in \Cref{fig:Cex1}.
It is similar to $M_4$ of \Cref{fig:Cex4}, and our arguments are thus similar to the proof of \Cref{prop:ClassicTimeVersusPureML}.
On $M_6$, we have $\Pmax{\Sched^\mathit{hist}_{\ell,v}}{} = 1$:
in $\ell_2$, the history reveals which of the two incoming edges was used, \ie which clock is already expired, thus the scheduler can make the optimal choice.
However, if neither the history nor $e$ is available, we get $\Pmax{\Sched^\mathit{ml}_{\ell,v}}{} = \frac{1}{2}$:
the only information that can be used in $\ell_2$ are the values of the clocks, but $v(x) = v(y)$, so there is no basis for an informed choice.
\end{proof}

\begin{proposition}
\label{prop:NonProphValVsTime}
$\Sched^\mathit{hist}_{\ell,t} \schedgtr \Sched^\mathit{ml}_{\ell,t}$ and $\Sched^\mathit{ml}_{\ell,v} \schedgtr \Sched^\mathit{ml}_{\ell,t}$.
\end{proposition}
\begin{proof}\proofnegspace{}
Consider the SA~$M_3$ in \Cref{fig:Cex3}.
We have $\Pmax{\Sched^\mathit{hist}_{\ell,t}}{} = \Pmax{\Sched^\mathit{ml}_{\ell,v}}{} = 1$, but $\Pmax{\Sched^\mathit{ml}_{\ell,t}}{} = \frac{1}{2}$ by the same arguments as in the proof of \Cref{prop:ClassicValVsTime}.
\end{proof}

\begin{proposition}
\label{prop:NonProphTimeVersusPureML}
$\Sched^\mathit{ml}_{\ell,t} \schedgtr \Sched^\mathit{ml}_{\ell}$.
\end{proposition}
\begin{proof}\proofnegspace{}
Consider the SA $M_4$ in \Cref{fig:Cex4}.
The schedulers in $\Sched^\mathit{ml}_{\ell}$ have no information but the current location, so they cannot make an informed choice in $\ell_2$.
This and the simple loop-free structure of $M_4$ make it possible to analytically calculate the resulting probability:
$\Pmax{\Sched^\mathit{ml}_{\ell}}{} = \frac{17}{24} = 0.708\overline{3}$. 
If information about the global elapsed time $t$ in $\ell_2$ is available, however, the value of $x$ is revealed.
This allows making a better choice, \eg going to $\ell_3$ when $t \leq \frac{1}{2}$ and to $\ell_4$ otherwise, resulting in $\Pmax{\Sched^\mathit{ml}_{\ell,t}}{} \approx 0.771$ (statistically estimated with high confidence).
\end{proof}
We have thus established the hierarchy of non-prophetic schedulers shown in \Cref{fig:NonPropheticHierarchy}, where some relationships follow from the propositions by transitivity.

\section{Experiments}
\label{sec:Experiments}

We have built a prototype implementation of lightweight scheduler sampling for SA by extending the \toolset's~\cite{HH14} \textsc{modes} simulator, which already supports deterministic stochastic timed automata (STA~\cite{BDHK06}).
With some care, SA can be encoded into STA.
Using the original algorithm for MDP of~\cite{DLST15}, our prototype works by providing to the schedulers a discretised view of the continuous components of the SA's semantics, which, we recall, is a continuous-space MDP.
The currently implemented discretisation is simple:
for each real-valued quantity (the value $v(c)$ of clock $c$, its expiration time $e(c)$, and the global elapsed time~$t$), it identifies all values that lie within the same interval $[\frac{i}{n}, \frac{i+1}{n})$, for integers $i$, $n$.
We note that better static discretisations are almost certainly possible, \eg a region construction for the clock values as in~\cite{KNSS00}.

We have modelled $M_1$ through $M_6$ as STA in \modest.
For each scheduler class and model in the proof of a proposition, and discretisation factors $n \in \set{1, 2, 4}$, we sampled $10\,000$ schedulers and performed statistical model checking for each of them in the lightweight manner.
In \Cref{fig:SMCResults} we report the min.\ and max.\ estimates, $(\hat p_\mathrm{min}, \hat p_\mathrm{max})_{\ldots}$, over all sampled schedulers.
Where different discretisations lead to different estimates, we report the most extremal values.
The subscript denotes the discretisation factors that achieved the reported estimates.
The analysis for each sampled scheduler was performed with a number of simulation runs sufficient for the overall max./min.\ estimates to be within $\pm\,0.01$ of the true maxima/minima of the \emph{sampled} set of schedulers with probability $\geq0.95$~\cite{DLST15}.
Note that $\hat p_\mathrm{min}$ is an upper bound on the true minimum probability and $\hat p_\mathrm{max}$ is a lower bound on the true maximum probability.

\begin{figure}[t]
\begin{multicols}{4}
\fontsize{8}{10}\selectfont
$M_1$:\\[2pt]
Proposition~\ref{prop:IncompClassicOrderedHist}:\\
$\Sched^\mathit{hist}_{\ell,v,e}\colon (0.24, 0.76)_{2,4}$\\
$\Sched^\mathit{hist}_{\ell,v,o}\colon (0.49, 0.51)_{1,2,4}$\\[4pt]
Proposition~\ref{prop:ClassicVsOrderVal}:\\
$\Sched^\mathit{ml}_{\ell,v,e}\colon (0.24, 0.76)_{2,4}$\\
$\Sched^\mathit{ml}_{\ell,v,o}\colon (0.49, 0.51)_{1,2,4}$\\[4pt]
Proposition~\ref{prop:ClassicVsOrderTime}:\\
$\Sched^\mathit{ml}_{\ell,t,e}\colon (0.24, 0.76)_{2,4}$\\
$\Sched^\mathit{ml}_{\ell,t,o}\colon (0.49, 0.51)_{1,2,4}$\\[4pt]
Proposition~\ref{prop:ClassicVsOrderSimple}:\\
$\Sched^\mathit{ml}_{\ell,e}\colon (0.24, 0.76)_{2,4}$\\
$\Sched^\mathit{ml}_{\ell,o}\colon (0.49, 0.51)_{1,2,4}$
\vfill\null
\columnbreak
$M_3$:\\[2pt]
Proposition~\ref{prop:ClassicValVsTime}:\\
$\Sched^\mathit{hist}_{\ell,t,e}\colon (0.00, 1.00)_1$\\
$\Sched^\mathit{ml}_{\ell,v,e\!}\colon (0.22, 0.78)_2$\\
$\Sched^\mathit{ml}_{\ell,t,e}\colon (0.40, 0.60)_2$\\[4pt]
Proposition~\ref{prop:ClassicVsOrderTime}:\\
$\Sched^\mathit{ml}_{\ell,t,e}\colon (0.40, 0.60)_2$\\
$\Sched^\mathit{ml}_{\ell,t,o}\colon (0.00, 1.00)_{1}$\\[4pt]
Proposition~\ref{prop:ClassicVsOrderSimple}:\\
$\Sched^\mathit{ml}_{\ell,e}\colon (0.38, 0.63)_2$\\
$\Sched^\mathit{ml}_{\ell,o}\colon (0.00, 1.00)_{1,2,4}$\\[4pt]
Proposition~\ref{prop:NonProphValVsTime}:\\
$\Sched^\mathit{hist}_{\ell,t}\colon (0.00, 1.00)_{1,2}$\\
$\Sched^\mathit{ml\phantom{k}}_{\ell,v}\colon (0.22, 0.78)_{4}$\\
$\Sched^\mathit{ml\phantom{k}}_{\ell,t}\colon (0.49, 0.51)_{1,2,4}$
\vfill\null
\columnbreak
$M_2$:\\[2pt]
Proposition~\ref{prop:OrderedHistVsMemoryless}:\\
$\Sched^\mathit{hist}_{\ell,o\phantom{,v}}\colon (0.06, 0.94)_{1,2,4}$\\
$\Sched^\mathit{ml}_{\ell,v,o}\colon (0.18, 0.83)_1$\\[8pt]
$M_4$:\\[2pt]
Proposition~\ref{prop:ClassicTimeVersusPureML}:\\
$\Sched^\mathit{ml}_{\ell,t,e}\colon (0.25, 0.79)_1$\\
$\Sched^\mathit{ml}_{\ell,e\phantom{,t}}\colon (0.29, 0.71)_1$\\[4pt]
Proposition~\ref{prop:NonProphTimeVersusPureML}:\\
$\Sched^\mathit{ml}_{\ell,t}\colon (0.22, 0.78)_{2,4}$\\
$\Sched^\mathit{ml}_{\ell}\colon (0.28, 0.72)_{1,2,4}$
\vfill\null
\columnbreak
$M_5$:\\[2pt]
Proposition~\ref{prop:WithinOrder}:\\
$\Sched^\mathit{ml}_{\ell,t,o}\colon (0.15, 0.86)_{4}$\\
$\Sched^\mathit{ml}_{\ell,o\phantom{,t}}\colon (0.16, 0.84)_{1,2,4}$\\[8pt]
$M_6$:\\[2pt]
Proposition~\ref{prop:NonProphHistVsFullMemoryless}:\\
$\Sched^\mathit{hist}_{\ell,v}\colon (0.00, 1.00)_{1,2}$\\
$\Sched^\mathit{ml\phantom{k}}_{\ell,v}\colon (0.49, 0.51)_{1,2,4}$
\end{multicols}\vspace{-15pt}
\caption{Results from the prototype of lightweight scheduler sampling for SA}
\label{fig:SMCResults}
\vspace{-0.5em}
\end{figure}

Increasing the discretisation factor or increasing the scheduler power generally increases the number of decisions the schedulers {\em can} make.
This may also increase the number of {\em critical} decisions a scheduler {\em must} make to achieve the extremal probability.
Hence,  the sets of discretisation factors associated to specific experiments may be informally interpreted in the following way:
\begin{itemize}
\item$\{1,2,4\}$: Fine discretisation is not important for optimality and optimal schedulers are not rare.
\item$\{1,2\}$: Fine discretisation is not important for optimality, but increases rarity of optimal schedulers.
\item$\{2,4\}$: Fine discretisation is important for optimality, optimal schedulers are not rare.
\item$\{1\}$: Optimal schedulers are very rare.
\item$\{2\}$: Fine discretisation is important for optimality, but increases rarity of schedulers.
\item$\{4\}$: Fine discretisation is important for optimality and optimal schedulers are not rare.
\end{itemize}
The results in~\Cref{fig:SMCResults} respect and differentiate our hierarchy.
In most cases, we found schedulers whose estimates were within the statistical error of calculated optima or of high confidence estimates achieved by alternative statistical techniques.
The exceptions involve $M_3$ and $M_4$.
We note that these models make use of an additional clock, increasing the dimensionality of the problem and potentially making near-optimal schedulers rarer.
The best results for $M_3$ and classes $\Sched^\mathit{ml}_{l,v,e},\Sched^\mathit{m,l}_{l,t,e},\Sched^\mathit{ml}_{l,e}$ were obtained using discretisation factor $n=2$: a compromise between nearness to optimality and rarity.
A greater compromise was necessary for $M_4$ and classes $\Sched^\mathit{ml}_{l,t,e},\Sched^\mathit{m,l}_{l,e}$, where we found near-optimal schedulers to be very rare and achieved best results using discretisation factor $n=1$.


The experiments demonstrate that lightweight scheduler sampling can produce useful and informative results with SA.
The present theoretical results will allow us to develop better abstractions for SA and thus to construct a refinement algorithm for efficient lightweight verification of SA that will be applicably to realistically sized case studies.
As is, they already demonstrate the importance of selecting a proper scheduler class for efficient verification, and that restricted classes are useful in planning scenarios.

\section{Conclusion}
\label{sec:Conclusion}

We have shown that the various notions of information available to a scheduler class, such as history, clock order, expiration times or overall elapsed time, almost all make distinct contributions to the power of the class in SA.
Our choice of notions was based on classic scheduler classes relevant for other stochastic models, previous literature on the character of nondeterminism in and verification of SA, and the need to synthesise simple schedulers in planning.
Our distinguishing examples clearly expose how to exploit each notion to improve the probability of reaching a goal.
For verification of SA, we have demonstrated the feasibility of lightweight scheduler sampling, where the different notions may be used to finely control the power of the lightweight schedulers.
To solve stochastic timed planning problems defined via SA, our analysis helps in the case-by-case selection of an appropriate scheduler class that achieves the desired tradeoff between optimal probabilities and ease of implementation of the resulting plan.

We expect the arguments of this paper to extend to steady-state/frequency measures (by adding loops back from absorbing to initial states in our examples), and that our results for classic schedulers transfer to SA with delayable actions.
We propose to use the results to develop better abstractions for SA, 
the next goal being a refinement algorithm for efficient lightweight verification of SA.


\begin{thebibliography}{10}
\providecommand{\url}[1]{\texttt{#1}}
\providecommand{\urlprefix}{URL }

\bibitem{deA99}
de~Alfaro, L.: {The Verification of Probabilistic Systems Under Memoryless
  Partial-Information Policies is Hard}. Tech. rep., DTIC Document (1999)

\bibitem{ACD91}
Alur, R., Courcoubetis, C., Dill, D.L.: Model-checking for probabilistic
  real-time systems (extended abstract). In: {ICALP}. LNCS, vol. 510, pp.
  115--126. Springer (1991)

\bibitem{AY06}
Andel, T.R., Yasinsac, A.: On the credibility of {MANET} simulations. {IEEE}
  Computer  39(7),  48--54 (2006)

\bibitem{ACGHHKMR15}
Avritzer, A., Carnevali, L., Ghasemieh, H., Happe, L., Haverkort, B.R.,
  Koziolek, A., Menasch{\'{e}}, D.S., Remke, A., Sarvestani, S.S., Vicario, E.:
  Survivability evaluation of gas, water and electricity infrastructures.
  Electr. Notes Theor. Comput. Sci.  310,  5--25 (2015)

\bibitem{BK08}
Baier, C., Katoen, J.P.: Principles of Model Checking. MIT Press (2008)

\bibitem{BBHPV13}
Ballarini, P., Bertrand, N., Horv{\'{a}}th, A., Paolieri, M., Vicario, E.:
  Transient analysis of networks of stochastic timed automata using stochastic
  state classes. In: {QEST}. LNCS, vol. 8054, pp. 355--371. Springer (2013)

\bibitem{BGHKNS16}
Bisgaard, M., Gerhardt, D., Hermanns, H., Krc{\'{a}}l, J., Nies, G., Stenger,
  M.: Battery-aware scheduling in low orbit: The {GomX-3} case. In: {FM}. LNCS,
  vol. 9995, pp. 559--576. Springer (2016)

\bibitem{BDHK06}
Bohnenkamp, H.C., D'Argenio, P.R., Hermanns, H., Katoen, J.: {MoDeST}: A
  compositional modeling formalism for hard and softly timed systems. {IEEE}
  Trans. Software Eng.  32(10),  812--830 (2006)

\bibitem{BD04}
Bravetti, M., D'Argenio, P.R.: Tutte le algebre insieme: Concepts, discussions
  and relations of stochastic process algebras with general distributions. In:
  Validation of Stochastic Systems. LNCS, vol. 2925, pp. 44--88. Springer
  (2004)

\bibitem{BG02}
Bravetti, M., Gorrieri, R.: The theory of interactive generalized semi-{M}arkov
  processes. Theor. Comput. Sci.  282(1),  5--32 (2002)

\bibitem{BKKR11}
Br{\'{a}}zdil, T., Krc{\'{a}}l, J., Kret{\'{\i}}nsk{\'{y}}, J., Reh{\'{a}}k,
  V.: Fixed-delay events in generalized semi-{M}arkov processes revisited. In:
  {CONCUR}. LNCS, vol. 6901, pp. 140--155. Springer (2011)

\bibitem{BBD03}
Bryans, J., Bowman, H., Derrick, J.: Model checking stochastic automata. {ACM}
  Trans. Comput. Log.  4(4),  452--492 (2003)

\bibitem{BKS14}
Buchholz, P., Kriege, J., Scheftelowitsch, D.: Model checking stochastic
  automata for dependability and performance measures. In: {DSN}. pp. 503--514.
  {IEEE} Computer Society (2014)

\bibitem{BHHK15}
Butkova, Y., Hatefi, H., Hermanns, H., Krc{\'{a}}l, J.: Optimal continuous time
  {M}arkov decisions. In: {ATVA}. LNCS, vol. 9364, pp. 166--182. Springer
  (2015)

\bibitem{DHLS16}
D'Argenio, P.R., Hartmanns, A., Legay, A., Sedwards, S.: Statistical
  approximation of optimal schedulers for probabilistic timed automata. In:
  {iFM}. LNCS, vol. 9681, pp. 99--114. Springer (2016)

\bibitem{DK05}
D'Argenio, P.R., Katoen, J.P.: A theory of stochastic systems part {I:}
  stochastic automata. Inf. Comput.  203(1),  1--38 (2005)

\bibitem{DArLM16}
D'Argenio, P.R., Lee, M.D., Monti, R.E.: Input/output stochastic automata -
  compositionality and determinism. In: {FORMATS}. LNCS, vol. 9884, pp. 53--68.
  Springer (2016)

\bibitem{DLST15}
D'Argenio, P.R., Legay, A., Sedwards, S., Traonouez, L.M.: Smart sampling for
  lightweight verification of {M}arkov decision processes. {STTT}  17(4),
  469--484 (2015)

\bibitem{EHZ10}
Eisentraut, C., Hermanns, H., Zhang, L.: On probabilistic automata in
  continuous time. In: {LICS}. pp. 342--351. {IEEE} Computer Society (2010)

\bibitem{GD07}
Giro, S., D'Argenio, P.R.: Quantitative model checking revisited: Neither
  decidable nor approximable. In: {FORMATS}. LNCS, vol. 4763, pp. 179--194.
  Springer (2007)

\bibitem{HS87}
Haas, P.J., Shedler, G.S.: Regenerative generalized semi-{M}arkov processes.
  Communications in Statistics. Stochastic Models  3(3),  409--438 (1987)

\bibitem{HHH14}
Hahn, E.M., Hartmanns, A., Hermanns, H.: Reachability and reward checking for
  stochastic timed automata. In: {AVoCS}. Electr. Comm. of the {EASST}, vol.~70
  (2014)

\bibitem{HS00}
Harrison, P.G., Strulo, B.: {SPADES} -- a process algebra for discrete event
  simulation. J. Log. Comput.  10(1),  3--42 (2000)

\bibitem{HH14}
Hartmanns, A., Hermanns, H.: The {M}odest {T}oolset: An integrated environment
  for quantitative modelling and verification. In: {TACAS}. LNCS, vol. 8413,
  pp. 593--598. Springer (2014)

\bibitem{HHK16}
Hartmanns, A., Hermanns, H., Krc{\'{a}}l, J.: Schedulers are no prophets. In:
  Semantics, Logics, and Calculi. LNCS, vol. 9560, pp. 214--235. Springer
  (2016)

\bibitem{DHS17}
Hartmanns, A., Sedwards, S., D'Argenio, P.: Efficient simulation-based
  verification of probabilistic timed automata. In: {WSC}. IEEE (2017), to
  appear.

\bibitem{Her02}
Hermanns, H.: Interactive {M}arkov Chains: The Quest for Quantified Quality,
  LNCS, vol. 2428. Springer (2002)

\bibitem{HKKS16}
Hermanns, H., Kr{\"{a}}mer, J., Krc{\'{a}}l, J., Stoelinga, M.: The value of
  attack-defence diagrams. In: {POST}. LNCS, vol. 9635, pp. 163--185. Springer
  (2016)

\bibitem{KCC05}
Kurkowski, S., Camp, T., Colagrosso, M.: {MANET} simulation studies: the
  incredibles. Mobile Computing and Communications Review  9(4),  50--61 (2005)

\bibitem{KNSS00}
Kwiatkowska, M.Z., Norman, G., Segala, R., Sproston, J.: Verifying quantitative
  properties of continuous probabilistic timed automata. In: {CONCUR}. LNCS,
  vol. 1877, pp. 123--137. Springer (2000)

\bibitem{LST15b}
Legay, A., Sedwards, S., Traonouez, L.M.: Estimating rewards \& rare events in
  nondeterministic systems. In: {AVoCS}. Electr. Comm. of the {EASST}, vol.~72
  (2015)

\bibitem{LST15a}
Legay, A., Sedwards, S., Traonouez, L.M.: Scalable verification of {M}arkov
  decision processes. In: {SEFM}. LNCS, vol. 8938, pp. 350--362. Springer
  (2015)

\bibitem{Mat62}
Matthes, K.: Zur {T}heorie der {B}edienungsprozesse. In: 3rd Prague Conf. on
  Information Theory, Stat. Dec. Fns. and Random Processes. pp. 513--528 (1962)

\bibitem{NS3Website}
{NS-3 Consortium}: ns-3: a discrete-event network simulator for internet
  systems. \url{https://www.nsnam.org/}

\bibitem{Pon93}
Pongor, G.: {OMNeT}: Objective modular network testbed. In: {MASCOTS}. pp.
  323--326. The Society for Computer Simulation (1993)

\bibitem{RS16}
Ruijters, E., Stoelinga, M.: Better railway engineering through statistical
  model checking. In: {ISoLA}. LNCS, vol. 9952, pp. 151--165. Springer (2016)

\bibitem{SZG12}
Song, L., Zhang, L., Godskesen, J.C.: Late weak bisimulation for {M}arkov
  automata. CoRR  abs/1202.4116 (2012)

\bibitem{Str93}
Strulo, B.: Process algebra for discrete event simulation. Ph.D. thesis,
  Imperial College of Science, Technology and Medicine. University of London
  (October 1993)

\bibitem{WBM06}
Wolf, V., Baier, C., Majster-Cederbaum, M.E.: Trace semantics for stochastic
  systems with nondeterminism. Electr. Notes Theor. Comput. Sci.  164(3),
  187--204 (2006)

\bibitem{Wol12}
Wolovick, N.: Continuous probability and nondeterminism in labeled transition
  systems. Ph.D. thesis, Universidad Nacional de C{\'o}rdoba, C{\'o}rdoba,
  Argentina (2012)

\bibitem{WJ06}
Wolovick, N., Johr, S.: A characterization of meaningful schedulers for
  continuous-time {M}arkov decision processes. In: {FORMATS}. LNCS, vol. 4202,
  pp. 352--367. Springer (2006)

\bibitem{ZBG98}
Zeng, X., Bagrodia, R.L., Gerla, M.: Glomosim: A library for parallel
  simulation of large-scale wireless networks. In: {PADS}. pp. 154--161. {IEEE}
  Computer Society (1998)

\end{thebibliography}
\end{document}